\documentclass[A4paper,11pt]{article}
\usepackage{amsmath,amssymb,amsfonts}

\newtheorem{theorem}{Theorem}[section]
\newtheorem{proposition}[theorem]{Proposition}

\newtheorem{definition}[theorem]{Definition}
\newtheorem{remark}[theorem]{Remark}

\newtheorem{example}[theorem]{Example}

\newenvironment{proof}{\mbox{\bf Proof.}}{\mbox{$\dashv$}\bigskip}

\begin{document}
    
    \begin{center}

        {\Large\bf   A Formal Axiomatization of Computation }  \\ ~\\

        {\bf  Rasoul Ramezanian}\\
        {\bf Ferdowsi University of Mashhad, Iran}\\ {rramezanian@um.ac.ir}\\

    \end{center}

\begin{abstract} 

\noindent We introduce an axiomatization  for the notion of computation. Based on the idea of Brouwer choice sequences, we construct a   model, denoted by $E$, which satisfies our axioms and $E \models \mathrm{ P \neq NP}$. 
In other words, regarding "effective computability" in Brouwer's intuitionism viewpoint, we show    $\mathrm{ P \neq NP}$.
\end{abstract}

\section{Introduction}
Is the famous problem $\mathrm{ P= NP}$ unprovable? To answer the question, we need an axiomatization for the notion of computation. In sections~2 and~3, we propose our setting and axiomatic system. 

To show that $\mathrm{ P = NP}$ is not derivable from our axiomatic systems,  we, in section~6, construct a model, denoted by $E$, which satisfies all of our axioms and does not satisfy  $\mathrm{ P = NP}$.

To construct our counter-model, $E$,  in section~4 and~5, we introduce non-predetermined functions (inspired by Brouwer choice sequences) and persistently evolutionary Turing machines as  an extension of Turing machines to compute non-predetermined functions. 

In computational complexity theory, the diagonal argument is used to show that two complexity classes are distinct. Nobody till now could be successful to use the diagonal argument to show $\mathrm{ P \neq NP}$.

In classic mathematics,  the diagonal argument is used to show that the size of the real numbers is  larger than the naturals. But in Brouwer intuitionism mathematics, instead of the diagonal argument, the principle of Bar induction~\cite{kn:TD} is used to show that there is no one-to-one correspondence between the natural numbers and the real numbers.

Our argument to show that $E \models \mathrm{ P \neq NP}$ is not the diagonal argument, and is in some sense similar to Bar induction principle. 
  
   In section~7, we argue that our axiomatic system plausibly formalizes ``natural computation"  similar to Peano axioms for ``natural numbers".

\section{Syntax and Definitions}
 A computation is a sequence of configurations that we transit from one to another by applying some instructions. The transitions are continued until a desired (an accept) configuration is obtained. In the following, we formally describe the notion of computation.

 Our syntax,  for explaining the notion of computation, consists of the followings
  
      \begin{itemize}
          \item[\emph{1.}] $\mathrm{INST}$ is a nonempty set called  the set of 
          \textit{``instructions"},
          
          \item[\emph{2.}] $\mathrm{CONF}$ is a nonempty set called the set of
          \emph{``configurations"} such that to each string $x\in \{0,1\}^*$,
          \begin{itemize}
              \item  a unique configuration $C_{0,x}\in \mathrm{CONF}$ is associated as ``the start
              configuration of the string $x$", and
              \item to each $C\in \mathrm{CONF}$,  a unique string $y_C\in
              \{0,1\}^*$ is associated as the string of involved in the configuration $C$, and also we have $y_{C_{0,x}}=x$ (see example~\ref{tce}).
          \end{itemize}

      \end{itemize}
  
      \begin{itemize}
          \item[\emph{3.}]  $\mathrm{TENG}$,  \emph{the ``transition   engine"}, is a total function
          from $\mathrm{CONF}\times\mathrm{ INST}$ to $\mathrm{CONF}\cup\{\bot\}$\footnote{ $\bot$ mean ``undefined".}.  
          
          \item[\emph{4.}] $\mathrm{AENG}$, \emph{the ``accepting   engine"},    is a total function
          from $\mathrm{CONF}$ to $\{YES, NO\}$.  
      \end{itemize}  
 For an example of the above syntax, one may see example~\ref{tce}.
  \begin{definition}\label{proced}
      ~\begin{itemize}
          
          \item[i.] \textbf{Procedures}.  A   \emph{procedure (an algorithm, a machine)}  is defined to be a finite  set
          $M\subseteq \mathrm{INST}$ \emph{(a finite set of instructions)}, satisfying the
          following condition 
          
          \begin{itemize} \item[-] \emph{The determination condition}: for
          every $C\in \mathrm{CONF}$,  at most there exists only one instruction $\tau$ in 
          $M$, denoted by notation  $\Upsilon(M,C)=\tau$, such that $\mathrm{TENG}(C,\tau)\in \mathrm{ CONF}$. If for all $\iota \in M$,
          $\mathrm{TENG}(C,\iota)=\bot$,   we denote $\Upsilon(M,C)=\emptyset$.
      \end{itemize}
          We refer to the
          set of all procedures by the symbol  $\Xi$.

           \item[ii.] \textbf{Languages}. A string $x\in  \{0,1\}^*$,  is in the
          \emph{language} of a  procedure $M$, denoted by $L(M)$, whenever we can
          construct a sequence $C_{0}C_{1},...,C_{n}$ of configurations in
          $\mathrm{CONF}$ such that
          \begin{itemize}\item $C_0=C_{0,x}$, \item each $C_i$, $i\geq 1$,
              is obtained by applying $\mathrm{TENG}$ on $(C_{i-1},\Upsilon(M,C_{i-1}))$,
              \item the $\mathrm{AENG}$ outputs $YES$ for $C_n$, and
              $\Upsilon(M,C_{n})=\emptyset$.
          \end{itemize}  The sequence $C_{0}C_{1},...,C_{n}$  is called the ``successful 
             computation path" of $M$ on $x$. The length of a
          computation path is the number of configurations appeared in.   
          
           \item[$iii.$] \textbf{Functions}. A partial function
          $f:\Sigma^*\rightarrow\Sigma^*$, $\Sigma=\{0,1\}$, is computed by a procedure $M\in
          \Xi$, whenever for $x\in \Sigma^*$,     we can construct a sequence
          $C_{0}C_{1},...,C_{n}$ of configurations in $\mathrm{CONF}$ such that
          \begin{itemize}\item $C_0=C_{0,x}$, \item each $C_i$, $i\geq 1$,
              is obtained by applying $\mathrm{TENG}$ on $(C_{i-1},\Upsilon(M,C_{i-1}))$,
              \item the $\mathrm{AENG}$ outputs $YES$ for $C_n$,   and 
              $\Upsilon(M,C_{n})=\emptyset$,
              
              \item $y_{C_n}=f(x)$.
          \end{itemize}
      
       \item[$iv.$] \textbf{Computation Path Length}. The   \emph{time
          complexity} of computing a  procedure $M$ on an input string
      $x$, denoted by $time_M(x)$,   is $n$, for some $n\in
      \mathbb{N}$, whenever   we can construct   a successful computation path
      of the  procedure $M$ on $x$ with length $n$.

      \item[$v.$] \textbf{Time Complexity}. Let $f:\mathbb{N}\rightarrow \mathbb{N}$  and
      $L\subseteq \Sigma^*$. The   time complexity
      of the computation of the language $L$  is less than $f$ whenever there exists a
      procedure  $M\in \Xi$ such that the language defined by the
      procedure $M$, i.e., $L(M)$, is equal to $L$, and for all
      $x\in L$, $time_M(x)<f(|x|)$.
      
       \item[$vi.$] \textbf{Complexity Classes.} We  define the time complexity class
      $\mathrm{P}\subseteq 2^{\Sigma^*}$ to be the set of all
      languages that we can compute   in polynomial time using $\mathrm{TENG}$ and $\mathrm{AENG}$. We
      also define the complexity class $\mathrm{NP}\subseteq
      2^{\Sigma^*}$ as follows:
      \begin{itemize}
          \item[] $L\in \mathrm{NP}$  iff there exists $J\in \mathrm{P}$
          and  a polynomial function $q$ such that for all $x\in
          \Sigma^*$,\begin{center} $x\in L\Leftrightarrow\exists y\in
              \Sigma^* (|y|\leq q(|x|) \wedge (x,y)\in J)$.\end{center}
      \end{itemize}
      \end{itemize} 
  \end{definition}
  
 \begin{remark}
 Definitions of computability and complexity classes stated in~\ref{proced} are not new and they are the same definitions appeared  in~\cite{kn:arora0} and~\cite{kn:mom0}.
 \end{remark} 

  \section{Axioms}
    
    In this section, we introduce the axioms of our setting. We only have 3 axioms.
      
      \begin{itemize} 
          
           \item[$A1.$] \textbf{Turing Computability and Complexity.} For every Turing machine $T$, there exists a procedure $M\in \Xi$ such that $L(M)=L(T)$ and the time complexity of $L(M)$ is equal to the time complexity of $L(T)$.

          \item[$A2.$] \textbf{Effective Computability of Engines }.   Both engines $\mathrm{TENG}$ and $\mathrm{AENG}$ are \textit{effectively computable}~(see~\ref{mechan}).
          \item[$A3.$]  \textbf{Time Complexity of Engines}.  Both engines $\mathrm{TENG}$ and $\mathrm{AENG}$ work in \textit{linear time}~(see~\ref{mechan2}).

      \end{itemize}
  In section~\ref{NCo}, we argue that these 3 axioms plausibly express the notion of ``natural computation". Axiom $A2$ and $A3$ reasonably express the attributes of the transition engine and the accepting engine.  We expect that both $\mathrm{TENG}$ and $\mathrm{AENG}$ are physically plausible and effectively computable, and just use linear time (clock) on configurations to determine the next configuration or accepting configurations.
  
\noindent In example~\ref{tce}, we introduce a   model, named $V$, which satisfies axioms $A1$, $A2$, and $A3$.
  
  \begin{example}\label{tce}       Let 
      \begin{itemize}
          \item[] $Q_T=\{h\}\cup\{q_i\mid   i\in \mathbb{N}\cup\{0\}\}$,
          \item[] $\Sigma,\Gamma$ be two finite set  with $\Sigma\subseteq \Gamma$
          and
          \item[]  $\Gamma$ has a symbol $\triangle \in \Gamma-\Sigma$.
      \end{itemize}

      \begin{itemize}
          \item[1)] $INST_v=\{[(q,a)\rightarrow(p,b,D)]\mid p,q\in Q_T,
          a,b\in \Gamma, D\in\{R,L\}\}$,

          
          \item[2)] $CONF_v=\{(q,x\underline{a}z)\mid q\in Q_T, x,z\in
          \Gamma^*, a\in \Gamma\}$,   for each $x\in \Sigma^*$,
          $C_{0,x}=(q_0,\underline{\triangle} x)$, and for each
          $C=(q,x\underline{a}z)\in CONF_v$, $y_C=xaz$.
          
      \end{itemize}
   
      \begin{itemize}
          
          \item[3)] Let $C=(q,xb_1\underline{a}b_2y)$ be an arbitrary
          configuration then
          
          \begin{itemize} \item $TENG_v(C,[(q,a)\rightarrow (p,c,R)])$ is defined to be $C'=(p,xb_1c\underline{b_2}y)$,
              
              \item $TENG_v(C,[(q,a)\rightarrow (p,c,L)])$ is defined to be
              $C'=(p,x\underline{b_1}cb_2y)$, and
              
              \item for other cases $TENG_v$ is defined to be $\bot$.
              
          \end{itemize} $TENG_v$ can be computed by a Turing machine in linear time.
          \item[4-] Let $C\in CONF_v$ be arbitrary
          \begin{itemize}
              \item if $C=(h,\underline{\triangle}x)$ then $AENG_v(C)$ is
              defined to be $YES$,
              
              \item if $C=(h,x\underline{\triangle})$ then $AENG_v(C)$ is
              defined to be $YES$, and
              
              \item otherwise $AENG_v(C)$ is defined to be $NO$.
          \end{itemize} $AENG_v$ can be computed by a Turing machine in linear time.
          
          \item[5-] For each $M\in \Xi_v$, and $C=(q,x\underline{a}y)\in
          CONF_v$,   if there exists $[(q,a)\rightarrow(p,b,D)]\in M$ for
          some $p\in Q_T, b\in \Gamma$, and $D\in \{R,L\}$, then
          $\Upsilon(M,C)=\{[(q,a)\rightarrow(p,b,D)]\}$ else
          it is defined to be  $\emptyset$.
      \end{itemize}
      
  \end{example} 
  
  \begin{remark}
  	  The model $V$ (example~\ref{tce}) is the same model of standard Turing machine which is recalled and  expressed by our proposed syntax. Every instruction $\tau \in \mathrm{INST_v}$  is a transition of Turing machines.  
  	  
  	 \noindent Suppose $T=\langle Q, \delta, \Sigma, \Gamma, F\rangle$ be a Truing machine, then in our syntax, the set $\delta$ is a procedure in the model $V$  that its language in the model $V$ is exactly the language of the Turing machine $T$. 
  \end{remark}

   It is obvious that $V$ satisfies axiom $A1$, $A2$, and $A3$, as both $TENG_v$ and $AENG_v$ are linear time Turing computable. 
  
  By axioms $A2$ and $A3$, it is required that the engines be linear time effective computable. In section 4 and 5, we discuss that effective Computability is not restricted to Turing Computability and introduce persistently evolutionary Turing machines.

\section{Non-predetermined Functions}
 \noindent The most important and fundamental notion of mathematics is function. A function is a process associating each element $x$ of a set $X$, to a single element $f(x)$ of another set $Y$. Classically, we assumed that all functions in mathematics are pre-determined.

\noindent In this section, we discuss functions that are not pre-determined and they are eventually determined through the way we start to associate $f(x)$ for every element $x\in X$.

\noindent We introduce Persistently Evolutionary Turing machines that compute non-predetermined functions.

Let $f$ be a process that associates elements of a set $X$ to the elements of another set $Y$. If the process $f$ works well-defined then we know $f$ as a mathematical function.
But being well-defined does not force the process $f$ to be predetermined. 

Suppose that $x_1$ and $x_2$ are two different elements of $X$. I want to use the process $f$ to determine the value of $f$ for $x_1$ and $x_2$. It is up to me to first perform the process $f$ on $x_1$ or $x_2$. 

If $f$ is predetermined the it does not matter to perform the process on ordering $x_1x_2$ or ordering $x_2x_1$. But if $f$ is non-predetermined then different order of inputs  causes different \textit{alternate functions} which one of them is the function that we are constructing.

 Alternate functions are functions that could exist in place of our function (if we interacted with different ordering of inputs, those alternate could happen).

For example, consider the following process $g$:
\begin{itemize}
    \item \textit{$W$ is a set which is initially empty.}
    
    \item \textit{for a given natural number $n$, if there exists a pair $(n,z) \in W$ then output $g(n)=z$, else update $W=W\cup\{(n,|W|+1)\}$ and output $g(n)=|W|+1$.}
\end{itemize}

The function $g$ is a non-predetermined function over natural numbers. I  input $7,9,1,11$ and the process will associates the following: $g(7)=1$, $g(9)=2$, $g(1)=3$, and $g(11)=4$. The value of other numbers are yet non-predetermined and as soon as I perform process $g$ on each number the value is determined.

\begin{itemize}
    \item[-] The  function $g$ is not predetermined. It is determined eventually, but it is always  undetermined for some numbers. 
    \item[-] The function $g$ is well-defined, and associates to each input a single output.
    \item[-] For every natural number, the function $g$ is definable. 
    \item[-] If I inputted  $9,1,7,11$, I would have an alternate $g$ which would associate: $g(9)=1$, $g(1)=2$, $g(7)=3$, and $g(11)=4$.
\end{itemize}

The inspiration of non-predetermined functions are Browser choice sequences~\cite{kn:TD}. A Choice sequence  is an   unfinished objects where the value of the sequence are not necessary predetermined. 
 \begin{quote}
 	A choice sequence is begun at a particular moment in time, and then grows as we choose further numbers. This process is generally open-ended and may be continued forever.~(page~89~of~\cite{kn:BH})
 \end{quote}
 
\section{Persistently Evolutionary Turing machines}
In this section, to have a formal computation model for non-predetermined functions,  we introduce the notion of Persistently Evolutionary Turing machines (we may also name them Brouwer-Turing machines) as an effective Computable method.  

Persistently Evolutionary Turing machines are an extension of the notion of Turing machines in which the structure of the machine can evolve through each computation.

A Turing machine consists of a set of states $Q$, and a table of transitions $\delta$ which both are fixed and remain unchanged forever. In Persistently Evolutionary Turing machines, we allow the set of states and the table of transitions changes through each computation. 

As a Persistently Evolutionary Turing Machine $PT$   computes on an input string $x$, the machine $PT$ can \underline{add} or \underline{remove} some of its states and transitions, and thus after the computation on the input $x$ is completed, the sets $Q$ and $\delta$ have been changed.
  
However, these changes are persistent. That is, if we already input a string $x$ and the machine outputs $y$, then whenever we again input $x$ the machine outputs the same $y$, and the changes of states and transitions do not violate well-definedness.

One may assume that we have a box and we set a Turing machine in the box with some rules of \underline{adding} and \underline{removing} of states and transitions. Then, We input strings to the box and for each string, the box outputs a single string. The machine in the box changes itself based on the rules, however, the behavior of the box is well-defined.
 
 Persistently Evolutionary Turing Machines computes non-predetermined functions. 
 
 \begin{definition} \label{perdef} An evolutionary Turing machine, $M=\langle Q, \delta,q_0, \Sigma=\{0,1\}, \Gamma=\Sigma\cup \{\bigtriangleup\}, F\rangle$~\footnote{$Q$: the set of states, $\delta$: the transition table, $F$: the set of accepting states} consists of the following:
     \begin{itemize}
         \item \textbf{The Initial Machine}. $M_0=\langle Q_0, \delta_0, F_0\rangle$, initially $M= M_0$.
         \item \textbf{Updating Rules}. There could be a finite set of updating rules. The following is the generic rule
         \begin{itemize}
         \item If during the computation, the machine $M$ on an input, say $x$, reached to a  configuration, say $C$, with a specific property, say $\phi$, then update 
         \begin{itemize} \item $\delta= \delta \cup \{t(C)\} $ or $\delta= \delta - \{t(C)\}$,
             \item $F= F \cup \{q(C)\} $ or $F=F - \{q(C)\}$,
                 \item $Q= Q \cup \{q(C)\} $ or $Q=Q - \{q(C)\} $,
             \end{itemize}
             
              where $t(C)$,  is a generic transition, and $q(C)$ is a generic state.

          \end{itemize}
          \item \textbf{Persistent}. Updating rules never violate well-definedness, and if $M$ accepts (rejects) an string $x$, whenever we apply $M$ on $x$ in future, again $M$ would accept (reject) $x$.
     \end{itemize}
     
\end{definition}

 In the following example, we introduce a persistently evolutionary nondeterministic finite automate~\cite{kn:mom0}.   

\begin{example}\label{autool}\emph{ (In the sequel of the paper, we will refer to the persistently evolutionary machine introduced in this example by $PT_{1}$). }
    
    \noindent Define $\mathrm{Evolve}:
    \mathrm{NFA}_1\times\Sigma^*\rightarrow \mathrm{NFA}_1$ as
    follows\footnote{$\mathrm{NFA}_1$ is the class of all
        nondeterministic finite automata   $M=\langle Q, \Sigma=\{0,1\}, \delta, q_0, F \rangle$,
        where for   each  state $q\in Q$, and $a\in \Sigma$, there
        exists at most one transition from $q$ with label $a$.}:
    
    \noindent Let $M\in\mathrm{NFA}_1$, $M=\langle Q,
    q_0,\Sigma=\{0,1\},\delta:Q\times\Sigma\rightarrow Q, F\subseteq
    Q\rangle\footnote{F is the set of accepting states}$,  and $x\in \Sigma^*$. Suppose $x=a_0a_1\cdots a_k$
    where $a_i\in \Sigma$. Applying the automata $M$ on $x$, one of
    the three following cases may happen:
    \begin{itemize}
        \item[case1.] The automata $M$   reads all $a_0,a_1\cdots ,a_k$
        successfully and stops in an accepting state. In this case,  the structure of the automata does not change   and let $\mathrm{Evolve}(M,x)=M$.
        
        \item[case2.] The automata $M$   reads all $a_0,a_1\cdots ,a_k$
        successfully and stops in a state $p$ which is not an accepting
        state.
        \begin{itemize} 
            \item  If the automata $M$ can  transit from the state $p$ to an  accepting
        state by reading only \emph{one} alphabet, then  let $\mathrm{Evolve}(M,x)=M$.
         \item If it cannot transit (from $p$ to an accepting state)  then let
        $\mathrm{Evolve}(M,x)$ to be a new automata $M'=\langle Q,
        q'_0,\Sigma=\{0,1\},\delta':Q'\times\Sigma\rightarrow Q',
        F'\subseteq Q'\rangle$, where $Q'=Q$, $\delta'=\delta$,
        $F'=F\cup\{p\}$.
    \end{itemize}
        \item[case3.] The automata $M$ cannot read all $a_0,a_1\cdots ,a_k$
        successfully,and after reading a part of $x$, say $a_0a_1\cdots
        a_i$, $0\leq i\leq k$, it crashes in  a state $q$ that
        $\delta(q,a_{i+1})$ is not defined. In this case, we let $\mathrm{Evolve}(M,x)$
        be a new automata $M'=\langle Q,
        q'_0,\Sigma=\{0,1\},\delta':Q'\times\Sigma\rightarrow Q',
        F'\subseteq Q'\rangle$, where $Q'= Q\cup \{s_{i+1},s_{i+2},\cdots,
        s_k\}$ (all $s_{i+1},s_{i+2},\cdots, s_k$ are new states
        that does not belong to $Q$), $\delta'=\delta\cup
        \{(q,a_{i+1},s_{i+1}), (s_{i+1},a_{i+2},s_{i+2}),$
        $\cdots,
        (s_{k-1},a_k,s_{k})\}$, and $F'=F\cup\{ s_k\}$.
    \end{itemize}
 
\end{example}
The machine $PT_{1}$ persistently evolve, that is, if it (rejected) accepted a string $x$ already, then it would (reject) accept the string $x$ for any future trials as well. The language $L(M)$ is not predetermined and it eventually is determined.

For example, assume that initially $M$ is $Q=\{q_0\}$, $F=\emptyset$, $\delta=\emptyset$. Now I input the string $101$ and according to case~3, the machine $M$ evolves and  new states $q_1,q_2,q_3$ and transitions $(q_0,1,q_1),(q_1,0,q_2),(q_2,1,q_3)$ are added and also $F=F\cup \{q_3\}$. Now if I input the string $10$ then according to case~2, $M$ rejects it. However, If at first I inputted $10$ to the machine then it would accept it.

\subsection{Time complexity of Evolutionary Turing machines}
The time-complexity~\cite{kn:arora0} of Persistent Evolutionary Turing Machines is defined similar to the time-complexity of Turing machines except that for each (\underline{adding})  \underline{removing} of states or transitions, we count one extra clock.

\begin{definition}
    Let  $M=\langle Q, \delta,q_0, \Sigma=\{0,1\}, \Gamma=\Sigma\cup \{\bigtriangleup\}, F\rangle$ be a persistently evolutionary Turing machine, $x$ is an arbitrary string, and  $f:\mathbb{N}\rightarrow \mathbb{N}$ is a function. We say that $time_M(x) < f(|x|)$, whenever everytime we compute $M$ on $x$, the sum of
    \begin{itemize}
        \item the number of uses of transitions in $\delta$, and
        \item the number of uses of updating rules 
        
        \end{itemize} happened during the computation $M$ on $x$,  is less than $f(|x|)$.
    
    \end{definition}

\begin{proposition}\label{timece} The time complexity of the machine $PT_1$ in example~\ref{autool} is linear.
\end{proposition}\begin{proof}
It is straightforward. 
\end{proof}

\subsection{Effective Computability}

In axioms $A2$ and $A3$, we required that both $\mathrm{TENG}$ and $\mathrm{AENG}$ to be effectively computable in linear time. In two following definitions, we formally explain what we mean by an effective Computability. 
 
 \begin{definition}\label{mechan}
  A function is  effectively computable if it is  either   Turing computable or   Persistently Evolutionary Turing computable. 
 \end{definition}

 \begin{definition}\label{mechan2}
     A function is computed in linear time whenever its corresponding Turing or Persistently Evolutionary Turing machine works in linear time. 
 \end{definition}

\section{  A Counter-Model } 

In this section, we prove that $\mathrm{P=NP}$ is not derivable from Axioms $A1$, $A2$ and $A3$. To do this, we construct a model $E$ which satisfies $A1$, $A2$ and $A3$ but not $\mathrm{P = NP}$.

\begin{definition}\label{Ee} We introduce a model $E$ as follows. \begin{itemize}

\item  Two sets $INST_e$ and $CONF_e$ are defined to be the same
$INST_v$ and $CONF_v$ in example~\ref{tce} respectively, and consequently the set   
  $\Xi_e$ is the
same $\Xi_v$.

\item  The transition engine $TENG_e$ is also defined similarly to the
transition engine $TENG_v$ in example~\ref{tce}, and thus it is linear time computable by a Turing machine.

\item  The accepting engine $AENG_e$ is defined as follows: let
$C\in CONF_e$ be arbitrary
\begin{itemize}
\item if $C=(h,\underline{\triangle}x)$ then $AENG_e(C)=YES$,

\item if $C=(h,x\underline{\triangle})$ then the $AENG_e$ works
exactly similar to the   the persistently evolutionary
 machine $PT_{1}$ introduced in example~\ref{autool}. On input $x$,  if $PT_1$ outputs $1$, the accepting engine
outputs $YES$, and

\item otherwise $AENG_e(C)=NO$.
\end{itemize}
By proposition~\ref{timece}, the engine $AENG_e$ is linear time computable by a persistently evolutionary Turing machine.  
\end{itemize}
\end{definition}

\begin{remark}
	 The only difference between model $E$ with model $V$ (see example~\ref{tce}) is that the $SBOX_e$ is a persistently evolutionary Turing machine.
\end{remark}

\begin{proposition}~\label{sati}
    The model $E$  satisfies axioms $A1$, $A2$, and $A3$.
\end{proposition}\begin{proof}
It is obvious by definition. 
\end{proof}

  Note that the $AENG_e$    is
a persistently evolutionary machine. The set of  procedures (algorithms) in the model
$E$ is the same set of  procedures in the model
$V$ (example~\ref{tce}), i.e $\Xi_v=\Xi_e$. However for some procedures, say $M$, the language $L(M)$ is the model $E$ could be different from the language $L(M)$ in $V$. For some $M\in \Xi_e$, we have $L(M)$ is a non-predetermined language. The procedure $M$ is fixed and does not change through time, but since the structure of  $AENG_e$   changes through time, the language $L(M)$ is non-predetermined.

\begin{definition}
We say a function $f:\mathbb{N}\rightarrow \mathbb{N}$ is
sub-exponential, whenever there exists $t\in \mathbb{N}$ such that
for all $n>t$, $f(n)<2^n$.
\end{definition}

\begin{theorem}\label{notmod} 
We have
\begin{center}
    $E\models \mathrm{P  \neq NP }$.
\end{center}
We show that there exists a  procedure $M\in \Xi_e$ such that

\begin{itemize}
\item the language $L(M)$ that the we  
compute  through $M$ is not predetermined,

\item the language $L(M)$  belongs to the class
$\mathrm{P}$, 

\item there exists no  procedure $M'\in \Xi_e$, such that
 $L(M')$  is equal to $$L'= \{x\in \Sigma^*\mid\exists y
(|y|=|x|\wedge y\in L(M))\},$$ and for some $k\in \mathbb{N}$, for
all $x\in L(M')$, if $|x|>k$ then
\begin{center}
$time_M(x)\leq f(|x|)$
\end{center}
where $f:\mathbb{N}\rightarrow \mathrm{N}$ is a sub-exponential
function. In other world, $L'$ is in $\mathrm{NP}$ but not in $\mathrm{P}$.

\end{itemize}
\end{theorem}\begin{proof}
 Consider the following procedure $M\in \Xi_e$
\begin{itemize} \item[] $\Sigma=\{0,1\},\Gamma=\{0,1,\triangle\},$

\item[]$M=\{[(q_0,\triangle)\rightarrow(h,\triangle,R)],
[(h,0)\rightarrow(h,0,R)],[(h,1)\rightarrow(h,1,R)]\}$.
\end{itemize} 
Using the $TENG_e$ and $AENG_e$, we compute procedure $M$ on an input $x$  as follows:
\begin{itemize}
	\item[-] Note $C_{0,x}=(q_0,\underline{\triangle} x)$, and   $\Upsilon(M,C_{0,x})=[(q_0,\triangle)\rightarrow(h,\triangle,R)]$.
	\item[-] We have $TENG_e(C_{0,x},[(q_0,\triangle)\rightarrow(h,\triangle,R)])=(h,\triangle\underline{a} x')$ where $x=ax'$ for some $a\in Sigma$.  
	\item[-] Continuing using the transition engine on the configurations, we reach to the configuration $(h,x\underline{\triangle})$ which $\Upsilon(M,(h,x\underline{\triangle}))=\emptyset$.
	\item[-] Running the accepting engine, $AENG_e(M,(h,x\underline{\triangle}))$, the persistently evolutionary NFA inside the accepting engine works and provide 'Yes' or 'No' as output.
\end{itemize}
The language of the   procedure $M$, $L(M)$, is not predetermined in model $E$. As   we  choose  a string $x\in
\Sigma^*$ to compute whether $x$ is an element of   $L(M)$, during the computation, the inner structure of the $AENG_e$ may evolve, and depending on the ordering of  strings, says $x_1,x_2,...$, that   we choose  to compute whether $x_i\in L(M)$, the language $L(M)$ eventually is determined.

 It is obvious that the language  $L(M)$ belongs to $\mathrm{P}$ 
(see the definition of time complexity in definition~\ref{proced}).

 Let $L'=\{x\in \Sigma^*\mid\exists y
(|y|=|x|\wedge y\in L(M))\}$. It is again obvious that $L'$
belongs to $\mathrm{NP}$ by definition~\ref{proced}.

We prove that $L'$ cannot belong to $\mathrm{P}$. Assume that $L'\in\mathrm{P}$ and thus
 there  exists  a  procedure $M'\in \Xi_e$ that we can compute $L'$ by $M'$ in time complexity less than a
sub-exponential function $f$. Then for some $k\in \mathbb{N}$, for
all $x$ with length greater than $k$, $x$ belongs to $L'$
whenever\begin{itemize} \item[] we construct  a
successful computation path $C_{0,x}C_{1,x},...,C_{n,x}$ of the
 procedure $M'$ on $x$, for some $n\leq
f(|x|)$.\end{itemize}
 Let $m_1\in \mathbb{N}$ be the maximum length of those
strings $y\in \Sigma^*$ that   \underline{until now} are accepted
by the persistently evolutionary   machine $PT_1$ (see example~\ref{autool}) which is
inside the $AENG_e$\footnote{Note that, since $m_1$ is the maximum length of those
	strings $y\in \Sigma^*$ that   \underline{until now} are accepted
	by the persistently evolutionary   machine $PT_1$, if we start to compute $PT_1$ for an arbitrary string, say $z$, with length greater than $m_1$ then $PT_1$ accepts $z$.}. 
\begin{center}
	Define $m=\max(m_1,k)$.
\end{center}

For every $y\in \Sigma^*$, let
$path(y):=C_{0,y}C_{1,y},...,C_{f(|y|),y}$ be the computational
path of  the procedure $M'$ on the string $y$. The $path(y)$ is  generated by the transition  engine (note that it is possible that the length of $path(y)$ is less that $f(|y|)$).  
 Let
\begin{center}
$S(y)=\{C_{j,y}\mid C_{j,y}\in path(y) \wedge \exists x\in
\Sigma^*(C_{i,y}=(h,x\underline{\triangle} ))\}$
\end{center}and
\begin{center}
$H(y)=\{x\in \Sigma^*\mid \exists C_{j,y}\in path(y)
(C_{j,y}=(h,x\underline{\triangle} ))\}$
\end{center}
The set $S(y)$ is the set of all configurations that the accepting engine on them runs its persistently evolutionary NFA, $PT_1$, during the computation $M'$ on the input string $y$. If $S(y)$ is empty then it means that  the computation $M'$ on $y$ does not force the structure of accepting engine, $AENG_e$, to  evolve. 

The set $H(y)$ is the set of all strings, say $z$,  that during the computation $M'$ on $y$, the persistently evolutionary NFA, $PT_1$ inside the accepting engine, works on $z$ as an input.

We refer by $|H(y)|$ to the number of elements of $H(y)$, and obviously we have
$|H(y)|\leq f(|y|)$ for $y$ with $|y|>k$. 

Also let $E(y)=H(y)\cap \{x\in
\Sigma^*\mid |x|=|y|\}$, and $D(y)=H(y)\cap \{x\in \Sigma^*\mid
|x|=|y|+2\}$.

\begin{center}
    \end{center}

\noindent Let $w\in \Sigma^*$ with $|w|>m$ be arbitrary. Two cases are possible:
either $S(w)=\emptyset$ or $S(w)\neq \emptyset$.

\underline{Consider the first case}. $S(w)=\emptyset$.

  Since the set $S(w)$ is empty, the
execution of $M'$ on $w$ does not make the $AENG_e$  evolve, and the value
$m_1$ remains unchanged. Here, we have also two cases:
\begin{itemize}
    \item[1)] If, using $TENG_e$ and $AENG_e$, we compute that $w$ is a member of $L(M')=L'$ then by definition of $L'$,  
    there must exist a string, say $v\in \Sigma^*$, such that
    
    \begin{center}
         $|v|=|w|$ and
        $v\in L(M)$ $(*)$.
    \end{center}  
     But is a contradiction. 
    \begin{itemize} 
        \item[(i)] Since we have the free will~\footnote{By free will, we mean that we are not forced to use $\mathrm{TENG}$ and $\mathrm{AENG}$ in any specific ordering.}, we first start to compute  procedure $M$ on all strings in $\Sigma^*$
        with length $|v|+1$ sequentially. As the length of $v$ is greater
        than $m$, all strings with length $|v|+1$ are accepted by the
        persistently evolutionary Turing machine $PT_1$ (see item-3 of
        example~\ref{autool}) which is inside  $AENG_e$. 
        
        \item[(ii)]    Then we   compute that whether $v$ is $L(M)$. But because of the evolution of $AENG_e$ happened in part (i), the  $AENG_e$ on computation of $M$ on $v$ outputs $NO$, and  thus
        $v$ is not an element of    $L(M)$ (see the item-2 of
        example~\ref{autool}). So $v\not\in L(M)$, and it contradicts with $(*)$.
    \end{itemize}

    \item[2)] If, using $TENG_e$ and $AENG_e$, we compute  that $w\not\in L(M')=L'$ then by definition of $L'$,   for all strings $v\in \Sigma^*$, $|v|=|w|$, we have $v\not\in
    L(M)$. But it contradicts with  the free will 
    again. As the length $w$ is greater than $m$, we  may
    choose a string $z$ with $|z|=|w|$ and by the item-3 of
    example~\ref{autool}, we have $z\in L(M)$, contradiction.
\end{itemize}

\underline{Consider the second case}. $S(w)\neq\emptyset$.

 Suppose that we, before computing $M'$ on $w$,   start to compute the procedure  $M$ on all strings $v0$'s, for all  $v\in E(w)$,
and then   compute procedure
$M$ on all strings $v0$'s, for all $v\in D(w)$ respectively.

 Since
$|w|>m$, we   have $u0\in L(M)$ for all $u\in
E(w)\cup D(w)$, and  $AENG_e$   evolves through computing $M$ on $u0$'s. It  evolves
in the way that   $AENG_e$ outputs $No$ for all
configuration in 
$$\{C_{i,w}\in S(w)\mid \exists x\in E(w)\cup
D(w)~s.t.~C_{i,w}=(h,x\underline{\triangle})\}.$$

 After that, we start to compute $M'$ on $w$. 
 Either we
finds $w\in L(M')$ or $w\not\in L(M')$.

\begin{itemize}
    \item  Suppose the first case
    happens and $w\in L(M')=L'$. It contradicts with the free will of us. We compute  $M$
    on all strings $v0$, $|v|=|w|$ sequentially, and would make $\{v0\in
    \Sigma^*\mid |v|=|w|\}\subseteq L(M)$. Then the $AENG_e$
    evolves in the way that, it will output $No$ for all
    configurations $(h,v\underline{\triangle})$, $|v|=|w|$, and   thus
    there would exist no $v\in L(M)\cap\{x\in \Sigma^*\mid |x|=|w|\}$ which implies $w\not\in L(M')$, contradiction.
    
    \item 
    Suppose the second case happens and $w\not\in L(M')=L'$. Since
    $|H(w)|<f(|w|)< 2^{|w|}$, during the computation of $M'$ on $w$,
    only $f(|w|)$ numbers of configurations  of the form
    $(h,x\underline{\triangle})$, $x\in\{v0\mid |v|=|w|\}\cup \{v1\mid
    |v|=|w|\}$ are given as input to the $AENG_e$. Therefore
    there exists a string $z\in \{x\in \Sigma^*\mid |x|=|w|\}$ such
    that none of its successors have been input to the persistently
    evolutionary Turing machine $PT_1$, and if we choose
    $z$ and computes $M$ on it, then $z\in L(M)$ which implies $w \in L'$. Contradiction.
\end{itemize}

 We showed that $L'$ cannot be computed by any $M'$ that its time complexity is less than a sub-exponential function. Thus $L'$ does not belong to the class $\mathrm{P}$. But because of the procedure $M$, we have $L'$ belongs to $\mathrm{NP}$ and therefore 
 
 \begin{center}
     $E\models \mathrm{P\neq NP}$.
 \end{center}
 
\end{proof}

The  above theorem simply says that if $L'$ belongs to $\mathrm{NP}$  then it forces us  to
 interact with $TENG_e$ and $AENG_e$   in some certain orders, which conflicts with our free will.
 
 \begin{theorem} ~
     
 \begin{center}
  $\mathrm{\{A1,A2,A3\}\not\vdash P=NP}$.
 \end{center}
\begin{proof}
It is a consequence of proposition~\ref{sati} and theorem~\ref{notmod}.
\end{proof}
 \end{theorem}


\section{Natural Computation} \label{NCo}
    
  For every mathematician, It is obvious that the set of ``natural numbers" is different from ``Peano axioms". In the same way, we can talk of ``natural computation" and our axiomatization setting.  
  
  So, one may ask
  \begin{itemize}\item  How much our setting with 3 axioms expresses the ``natural computation"?\item How the ``natural transition engine" of the reality works?
\item How the ``natural accepting engine" works?

\item Is our axiomatic system plausibly formalize the ``natural computation"?
  \end{itemize}
  The Church-Turing thesis states that
  \begin{quote}
  a function on the natural numbers can be calculated by an effective method, if and only if it is computable by a Turing machine.
  \end{quote}  
 
\noindent If we want to recall the Church-Turing thesis in our setting, it says
 \begin{quote}
     a function on the natural numbers can be calculated by an effective method, if and only if it is computable by a procedure $M$ in $\Xi$.
 \end{quote}

When we perform a computation, we transit from a configuration to another configuration (using $\mathrm{TENG}$ of the reality) and also, we check whether a configuration is accepted or not (using $\mathrm{AENG}$ of the reality).

We do not know what is the inner structure of  $\mathrm{TENG}$ and  $\mathrm{AENG}$ of the reality, but we believe that  both $\mathrm{TENG}$ and  $\mathrm{AENG}$ are physically plausible, and thus
 
\begin{itemize}
    \item[1.] both $\mathrm{TENG}$ and  $\mathrm{AENG}$ of the reality are effectively computable, and
    
    \item[2.] both $\mathrm{TENG}$ and  $\mathrm{AENG}$ work  in linear time.
\end{itemize}
We state these two properties in axioms $A2$ and $A3$.
   We, inhabitants of reality, can never find out whether the reality persistently evolves or not.   We can never discover that whether the $\mathrm{TENG}$ and $\mathrm{AENG}$  of the reality is a Turing machine or a Persistently Evolutionary Turing machine. 
   
   We believe that our setting and 3 axioms, plausibly formalize the ``natural computation" similar to Peano axioms for natural numbers.
   
   We cannot derive $\mathrm{ P = NP}$ from our 3 axioms which forces us to consider the engines of the reality to effectively compute in linear time.
   
   \section{Conclusion}
 We proposed an axiomatic system for ``natural computation". We justified that our axioms plausibly describe the ``natural computation"  similar to Peano axioms for natural numbers. We  show that $\mathrm{ P= NP}$ is not derivable from our axioms.
 
 We also show that regarding   "effective computability" from Brouwer's intuitionism viewpoint,  $\mathrm{ P \neq NP}$.
 
 \begin{itemize}
     \item []
     
 \end{itemize}

\noindent \textbf{Acknowledgment}. I would like to thank Prof. Amir Daneshgar for his valuable comments.

\end{document}